%% file: draftDec10.tex
\documentclass{amsart}
\usepackage{mathrsfs}
\usepackage{cases}
\usepackage{amssymb,amscd,amsthm}
\usepackage{latexsym}
\usepackage{hyperref}
\usepackage{color}
\usepackage{mathptmx}
\usepackage[noadjust]{cite}
\usepackage{mathtools}

\usepackage{subfigure}
\usepackage{graphicx}

\date{\today}

\newcommand{\Z}{{\mathbb Z}}
\newcommand{\R}{{\mathbb R}}

\newcommand{\N}{{\mathbb N}}

\newcommand{\hdim}{\dim_{\mathrm{H}}}

\newcommand{\tr}{\mathrm{tr}}

\newtheorem{theorem}{Theorem} [section]

\newtheorem{lemma}[theorem]{Lemma}
\newtheorem{prop}[theorem]{Proposition}
\newtheorem{coro}[theorem]{Corollary}

\DeclarePairedDelimiter\ceil{\lceil}{\rceil}
\DeclarePairedDelimiter\floor{\lfloor}{\rfloor}
   
\begin{document}

\title{Frequency dependence of H\"older continuity for quasiperiodic Schr\"odinger operators}
\author{Paul E. Munger}
\address{Rice University, Houston, TX}
\email{pem1@rice.edu}
\begin{abstract}
We prove estimates on the H\"older exponent of the density of states measure for discrete Schr\"odinger operators with potential of the form $V(n) = \lambda(\floor{(n+1)\beta} - \floor{n\beta})$, with $\lambda$ large enough, and conclude that for almost all values of $\beta$, the density of states measure is not H\"older continuous.
\end{abstract}
\maketitle

\section{Introduction}

The discovery of quasicrystals by Dan Shechtman \cite{SBGC} has elicited considerable interest (for instance \cite{T}, \cite{BM}) in the subject of aperiodic order from mathematicians and physicists. The Fibonacci Hamiltonian has been one of the canonical models for a quasicrystal. It is the one-dimensional discrete Schr\"odinger operator on $\ell^2(\Z)$ specified by
$$(H\psi)(n) = \psi(n-1) + V(n)\psi(n) + \psi(n+1).$$
The sequence $V$ is called the \textit{potential}; for the Fibonacci Hamiltonian, $V(n) = \lambda(\floor{(n+1)\beta} - \floor{n\beta})$, where $\lambda >0$ is called the \textit{coupling\ constant} and $\beta = \frac{\sqrt{5} - 1}{2}$ the \textit{frequency}.

The spectral properties of the Fibonacci Hamiltonian qualitatively agree with those of physical quasicrystals. For example,
\begin{enumerate}

\item The spectrum $\Sigma$ is a Cantor set of zero Lebesgue measure (\cite{Su}). Its point spectrum is empty (\cite{Su}, \cite{DL}), so the spectral measure is purely singular continuous. As $\lambda\to\infty$, the Hausdorff dimension of the spectrum behaves like $1.831 \cdot \frac{-\log{\beta}}{\log\lambda}$(\cite{DEGT}).

\item The spectral measure $\mu$ is uniformly $\alpha$-H\"older continuous for some $\alpha > 0$. This means that there is a $\delta > 0$ such that for all $x$ and $y$ with $|x-y| < \delta$, $\mu[x,y] < |x-y|^\alpha$ (\cite{DL}). An asymptotically optimal estimate of $\alpha$ has not been established, but see \cite{DG11}.

\item The \textit{density of states measure} $N$ is also $\alpha$-H\"older continuous (\cite{DG},\cite{DKL00}). The distribution function of the density of states measure is given by the formula (\cite{H})
$$N([x,y]) = \lim_{n\to\infty} \frac{\#\{\mathrm{eigenvalues\ of\ } H_n \mathrm{\ in\ } [x,y]\}}{n},$$
where $H_n$ is the restriction of $H$ to the $\ell^2$ sequences supported on $[1,n]$. As $\lambda \to \infty$, the optimal H\"older exponent behaves like $\frac{-3 \log\beta}{2\log\lambda}$.

\end{enumerate}

One wonders how these properties depend on the frequency. It is already known (\cite{BIST}) that $\Sigma$ is a Cantor set of zero Lebesgue measure for all irrational values of $\beta$. For $\lambda$ large enough, \cite{LW} established estimates for the Hausdorff dimension of $\Sigma$. Let $[0;a_1,a_2,a_3,\dots]$ be the continued fraction expansion of $\beta$ (see \cite{Kh} for an introduction to continued fractions). Then when $\underline{M}(\beta) := \liminf_{k\to\infty} \sqrt[k]{a_1\dots a_k}$ is finite, 
\begin{multline*}
\max\left\{\frac{\log 2}{10 \log 2 + 3\log(4(\lambda-8))}, \frac{\log M(\beta) - \log 3}{\log M(\beta) + \log(12(\lambda-8))} \right\} \le \hdim (\Sigma)\\
\le \frac{2\log M(\beta) + \log 3}{2\log M(\beta) + \log(\lambda - 8) - \log 3}.
\end{multline*}
When $M$ is infinite, the Hausdorff dimension of the spectrum is $1$; notice that the upper bound depends on $\lambda$ in the same way for all $\beta$. In \cite{LQW}, Liu, Qu, and Wen derive an expression for the Hausdorff dimension of $\Sigma$ and show that for all $\beta$ and $\lambda > 24$, the Hausdorff dimension is Lipschitz continuous.

Using methods like those of \cite{DG}, this article determines how the $\alpha$-continuity of the density of states measure depends on $\beta$, assuming throughout that $\overline{M}(\beta) = \limsup _{k\to\infty} \sqrt[k]{a_1 a_2 \dots a_k} < \infty$ and $\lambda > 24$ (recall that $\overline{M}$ is finite for almost all $\beta$). 

When the continued fraction coefficients of $\beta$ are constant, the behavior is like the Fibonacci Hamiltonian:

\begin{theorem}\label{hcont-bb}
Suppose $\beta = [0;b,b,b,\dots]$. Then for every
$$\gamma < \begin{dcases}
\frac{2 \log\beta}{-b\log(\lambda+5) - 3\log(b+2)} &b > 3 \\
\frac{\log\beta}{-\log(\lambda+5) - 3\log(b+2)} &b=2,3 \\
\frac{3\log\beta}{-2\log(27(\lambda+5))} &b=1
\end{dcases},$$
there is a $\delta > 0$ such that the density of states measure $N$ associated to the family of Schr\"odinger operators with frequency $\beta$ and coupling strength $\lambda$ obeys
$$|N(x) - N(y)| \le |x-y|^\gamma$$
for all $x$, $y$ with $|x-y| < \delta$.
\end{theorem}

\begin{theorem}\label{opt-bb}
If $\beta = [0;b,b,b,\dots]$, then for every
$$ \tilde\gamma > \begin{dcases}
\frac{2\log\beta}{-b\log(\lambda-8) - \log(b) + b\log 3} &b > 2 \\
\frac{\log\beta}{-\log(\lambda-8) + \log(b) - \log 3} &b=2 \\
\frac{3\log\beta}{-2\log(\lambda-8) + - 2\log 3} &b=1
\end{dcases},$$
and any $0<\delta<1$, there are $x$ and $y$ with $0<|x-y|<\delta$ such that $|N(x) - N(y)| \ge |x-y|^{\tilde\gamma}$.
\end{theorem}

\begin{coro}

For constant continued fraction coefficients, this identifies the optimal asymptotic behavior of $\gamma$ as $\lambda \to \infty$. If $\Gamma(\lambda, b)$ is the optimal H\"older exponent, $\gamma \le \Gamma \le \tilde\gamma$, so that
$$\lim_{\lambda\to\infty} \Gamma(\lambda,b)\log\lambda =  \frac{-2\log\beta}{b}$$
when $b > 3$, and similarly for smaller values of $b$.

\end{coro}

More generally, the qualitative behavior is determined by $\overline{d}(\beta) = \limsup_{N\to\infty} \frac{1}{N}\sum _{i=1} ^N a_i$ and $\underline{d}$, the limit inferior.
\begin{theorem}\label{hcont-d}
If $\overline{d}$ is finite, then $N$ is $\alpha$-H\"older continuous for some $\alpha$. If $\underline{d}$ is infinite, $N$ is not H\"older continuous. It is well known that $\overline{d} = \underline{d} = \infty$ almost everywhere. Thus, for Lebesgue almost all $\beta$, $N$ is not H\"older continuous.
\end{theorem}

\section{Structure of $\Sigma$}

In \cite{LW}, the fine structure of $\Sigma$ is developed enough (along the lines of \cite{R}) to estimate its Hausdorff dimension. This article uses many parts of the apparatus Liu and Wen develop, so we recapitulate the necessary results (without proof). The idea is to approximate $\Sigma$ by finite unions of closed intervals, growing in number and shrinking in size at controlled rates. 

The central objects to approximating $\Sigma$ are the continued fraction approximations to $\beta$. Let $p_k / q_k$ be the $k^{\mathrm{th}}$ convergent to $\beta$. For $k \ge 1$ and $x\in\R$ define the \textit{transfer matrix over $q_k$ sites} by
$$
M_k (E) = \prod_{n=q_k} ^1 \begin{bmatrix}
E - V(n) & -1 \\ 1 & 0
\end{bmatrix}
$$
and put
$$
M_{-1} (E) = \begin{bmatrix} 1 & -\lambda \\ 0 & 1\end{bmatrix},\ M_0 (E) = \begin{bmatrix}E & -1 \\ 1 & 0\end{bmatrix}.
$$
These matrices arise in the spectral theory of a discrete Schr\"odinger operator because they produce the sequences that satisfy the formal difference equation $H\psi = E\psi$.

\begin{prop}\label{basic-trace}
This summarizes work that first appeared in \cite{R}. Let $x_{(k,p)} = \tr M_{k-1} M_k ^p$ and $\sigma_{(k,p)} = \{E \in \R : |x_{(k,p)} (E) | \le 2\}$.  Then:
\begin{enumerate}
\item $M_{k+1} = M_{k-1} M_k^{a_k}$, so that $x_{(k+2,0)} = x_{(k,a_k)}$.
\item For $\lambda > 4$, $\sigma_{(k,p)}$ is made of disjoint closed intervals, equal in number to the degree of $x_{(k,p)}$. These intervals are called \textit{bands}.
\item $\sigma_{(k+2,0)} \cap \sigma_{(k+1,0)} \subset \sigma_{(k+1,0)} \cap \sigma_{(k,0)}$.
\item $\sigma(H) = \bigcap_k (\sigma_{(k+1,0)} \cup \sigma_{(k,0)})$.
\item If $k \in \N$ and $p \ge -1$, $\sigma_{(k,p+1)} \subset \sigma_{(k+1,0)} \cap \sigma_{(k,p)}$. 
\item If $k \in \N$, $p \ge 0$, and $\lambda > 4$, $\sigma_{(k+1,0)} \cap \sigma_{(k,p)} \cap \sigma_{(k,p-1)} = \emptyset$.
\end{enumerate}
\end{prop}

We will approximate $\Sigma$ using a certain subset of the above bands, called the \textit{generating bands}. For $k \in \N$, define:
\begin{enumerate}
\item A band of type $(k,I)$ is a band of $\sigma_{(k,1)}$ contained in a band of $\sigma_{(k,0)}$.
\item A band of type $(k,II)$ is a band of $\sigma_{(k+1,0)}$ contained in a band of $\sigma_{(k,-1)}$.
\item A band of type $(k,III)$ is a band of $\sigma_{(k+1,0)}$ contained in a band of $\sigma_{(k,0)}$.
\end{enumerate}
For each value of $k$, call the the collection of all bands of the above three kinds the \textit{spectral generating bands at level k}, written $\mathcal{G}_k$. These bands are useful because the combinatorial structure of $\mathcal{G}_{k+1}$ is easy to describe if $\mathcal{G}_k$ is known.

\begin{lemma}\label{band-comb}
For $k\in\N$,
\begin{enumerate}
\item Each band of type $(k,I)$ contains a single generating band; it is a band of $\sigma_{(k+2,0)}$ of type $(k+1,II)$.
\item Each band of type $(k,II)$ contains $a_k +1$ bands of $\sigma_{(k+1,1)}$ of type $(k+1,I)$, and $a_k$ bands of $\sigma_{(k+2,0)}$ of type $(k+1,III)$.
\item Each band of type $(k,III)$ contains $a_k$ bands of $\sigma_{(k+1,1)}$ of type $(k+1,I)$ and $a_k -1$ bands of $\sigma_{(k+2,0)}$ of type $(k+1,III)$.
\end{enumerate}
\end{lemma}

\begin{figure}\label{band-tree}
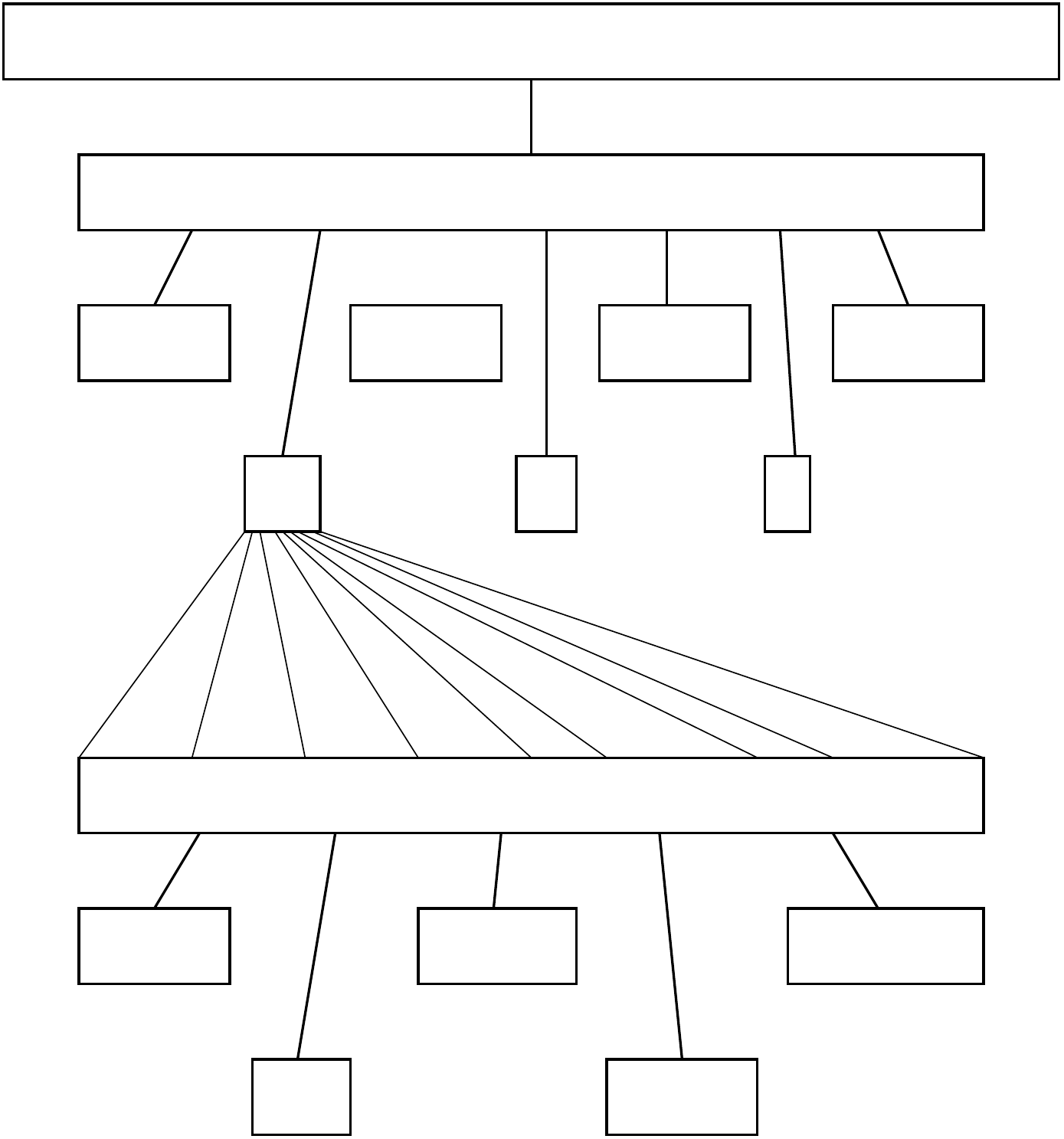
\caption{An illustration of Lemma \ref{band-comb} for $a_i\equiv 3$.}
\end{figure}

To make use of Lemma \ref{band-comb}, we have to understand how $\mathcal{G}_k$ approximates $\Sigma$:

\begin{lemma}\label{gen-spect}
Every generating band of level $k$ is contained in some generating band of level $k-1$, and $\Sigma = \bigcap_{k=0}^\infty \bigcup_{B\in\mathcal{G}_k} B$.
\end{lemma}

Let us recast the content of Lemma \ref{band-comb} in a form that is more useful for calculations. Put 
$$T_k = \begin{bmatrix}0 & 1 & 0 \\ a_k + 1 & 0 & a_k \\ a_k & 0 & a_k-1\end{bmatrix}.$$
Then Lemma \ref{band-comb} says that for $i,j\in\{I,II,III\}$, a band of type $(k,i)$ spawns $T(i,j)$ bands of type $(k+1,j)$. Given a band $B\in\mathcal{G}_k$, we can associate to it a \textit{type index}, which is the sequence of the types ($I$,$II$, or $III$) of each of its forbears. That is, for $B \in \mathcal{G}_k$, put $\tau(B) = i(0),i(1),\dots,i(k)$ where $B$ is of type $(k,i(k))$ and $B$ is contained in a band of type $(l,i(l))$ for each $l<k$. Using the matrices $T_k$, the number of bands with a given type index $\tau$ is seen to be $\prod_{l=0} ^{k-1} T (\tau(l),\tau(l+1))$. Define $\mathcal{A}$ to be the set of all one-sided infinite words $\tau$ on $\{I,II,III\}$ such that each prefix of $\tau$ is the type index of at least one band. Such type sequences are called \textit{admissible}.

Because of the self-similarity of $\Sigma$, estimates for the lengths of generating bands are amenable to a similar formalism. Put
$$
P_k = \begin{bmatrix}
0 & (3/(\lambda-8))^{a_k-1} & 0 \\
3/(a_k(\lambda - 8)) & 0 & 3/(a_k(\lambda - 8)) \\
 3/(a_k(\lambda - 8)) & 0 & 3/(a_k(\lambda - 8))  
\end{bmatrix},
$$
$$
Q_k = \begin{bmatrix}
0 & (1/(\lambda + 5))^{a_k-1} & 0\\
(\lambda + 5)^{-1}(a_k+2)^{-3} & 0 &(\lambda + 5)^{-1}(a_k+2)^{-3} \\
(\lambda + 5)^{-1}(a_k+2)^{-3} & 0 & (\lambda + 5)^{-1}(a_k+2)^{-3} \\
\end{bmatrix}.
$$

\begin{lemma}\label{bl1}
Then if $B$ is a generating band and $\tau$ its type,
$$4\prod Q_l (\tau(l),\tau(l+1)) \le |B| \le 4\prod P_l (\tau(l), \tau(l+1)).$$
\end{lemma}

\section{Band Length Estimates}

In this section we consider the asymptotic scaling rate of bands in $\sigma_k$. It is easy to treat the case $\beta = [0;b,b,b,\dots]$ quantitatively, so we take it up first. 

\begin{lemma}\label{ub-bb}
If $\beta = [0;b,b,b,\dots]$, then for all $k$ and $B \in \mathcal{G}_k$,
$$\log\frac{|B|}{4} \ge \begin{dcases}
-\ceil{\frac{k}{2}} (b-1) \log(\lambda+5) - 3\floor{\frac{k}{2}}\log(b+2) - \floor{\frac{k}{2}}\log(\lambda+5), &b>3 \\
-k\log(\lambda+5) - 3k\log(b+2), &b=2,3 \\
\floor{\frac{2k}{3}}\left(-\log(\lambda+5) - 3\log(b+2)\right), &b=1.
\end{dcases}
$$
Let $L(k)$ be the above bound (that is, $|B| \ge L(k)$), and write $L(b,\lambda)$ for $\liminf_{k\to\infty} \frac{1}{k}\log L(k)$.
\end{lemma}

\begin{proof}

Our task is to bound the lower bound in Lemma \ref{bl1} from below as $B$ ranges over all of $\mathcal{G}_k$. A look at the matrix $Q_k$ shows that if $b\ge 4$ the bound in question is minimized whenever $\tau(B)$ has the greatest admissible amount of $I$s in it. This can occur at most half the time (by Lemma \ref{band-comb}), proving the formula in the first case.

On the other hand, when $b < 4$ the minimum is achieved by a band with the greatest possible amount of $II$s and $III$s in its type index. When $b \ne 1$, this is possible for every entry of $\tau(B)$, giving the second case, and for $b=1$ it is possible $2/3$ of the time.

\end{proof}

\begin{lemma}\label{blub-bb}
If $\beta = [0;b,b,b,\dots]$, then for each $k$, there is a band of $\mathcal{G}_k$ with
$$\log\frac{|B|}{4} \le \begin{dcases}
-\ceil{\frac{k}{2}} (b-1)(\log(\lambda-8)-\log 3) - \floor{\frac{k}{2}}(\log b + \log(\lambda-8) - \log 3), &b>2 \\
-k(\log(\lambda-8) - \log(b) + \log 3), &b=2 \\
\frac{-2k}{3}(\log(\lambda-8) - \log(b) + \log 3), &b=1.
\end{dcases}
$$
Let $U(k)$ be the above bound and define $U(b,\lambda)=\limsup _{k\to\infty} \frac{1}{k}\log\frac{U(k)}{4}$.
\end{lemma}

\begin{proof}
The analysis is the same as for the previous lemma, substituting $P_k$ for $Q_k$.
\end{proof}

Now we consider arbitrary $a_i$s. In this case the optimization problem treated in the above estimates becomes intractable, so we focus on the qualitative behavior.

\begin{lemma}\label{nope-pd}
If $\underline{d}$ is infinite, the sequence of lengths of shortest bands of $\sigma_{k+1,0}$ decays to zero faster than any geometric sequence.
\end{lemma}

\begin{proof}
We argue roughly as in the proof of Lemma \ref{blub-bb}. Indeed, $\overline{d}=\infty$ means there is a $c>0$ such that in each $\sigma_{k+1,0}$ there is a band whose length is less than $\exp(-c k d_k)$ (where $d_k$ is the average of the first $k$ coefficients of $\beta$). Since $c$ is independent of $k$ and $d_k$ diverges, $\exp(-c k d_k)$ approaches zero faster than any geometric sequence.
\end{proof}

\begin{lemma}\label{ub-arb}
If $\overline{d}$ is finite, then there is a geometric sequence that bounds the length of every band in $\sigma_k$ from below.
\end{lemma}

\begin{proof}
Suppose not, so that for some faster-than-geometric sequence $l_k$, $\sigma_k$ has a band of length less than $l_k$. This forces the lower bound of \ref{bl1} to decay to zero faster than geometrically for some admissible type sequence. This means that $a_i$ has a subsequence $a_{i_n}$ such that $\frac{1}{k} \sum _{\{n : i_n \in [1,k]\}} \log c^{a_{i_n}} $ or $\frac{1}{k} \sum _{\{n : i_n \in [1,k]\}} \log a_{i_n}$ diverge. Both contradict $\overline{d} < \infty$ (and the latter also contradicts $\overline{M} < \infty$).
\end{proof}

\section{H\"older Continuity}

\begin{lemma}\label{qgrow}
Recalling that $\overline{M}(\beta) < \infty$, the sequence $q_k$ of denominators of convergents to $\beta$ is bounded above and below by geometric sequences.
\end{lemma}

\begin{proof}
By definition,
$$q_{k+1} = a_{k+1} q_k + q_{k-1}.$$
Since $q_k$ increases monotonically, it follows that $q_{k+1} \le (a_{k+1} + 1)q_k$. Thus
$$\log q_{k+1} \le \sum_{i=1} ^k \log(a_{i+1} + 1).$$
Since $\log(x+1) \le \log(x) + 1/x$ and $\sum _{i=1} ^k 1/a_i$ grows no faster than $k$, $\log q_k$ is bounded above by an arithmetic sequence.

Again using monotonicity of $q_k$, we get $q_{k+1} \ge a_{k+1} q_k$. Repeating the above reasoning, we finish the proof.
\end{proof}

\begin{proof} [Proof of Theorem \ref{hcont-bb}]

For arbitrary $x_0 < y_0$ that are close enough, we want to estimate $N(y_0) - N(x_0)$ from above. Because there is one Dirichlet eigenvalue of $H_{q_k}$ associated to each band of $\sigma_{k+1,0}$ (\cite{H}),
$$N(y_0) - N(x_0) = \lim_{n \to \infty} \frac{\# \sigma(H_n) \cap [x_0,y_0]}{n} = \lim_{k \to \infty} \frac{\# \sigma_{k+1,0} \cap [x_0,y_0]}{q_k},$$
where $\#X$ is the cardinality of $X$. This amounts to finding a bound on $\# \sigma_{k+1,0} \cap [x_0,y_0]$. With $L$ the bound on band lengths of Lemma \ref{ub-bb}, define $m$ by
$$L(m+1) \le y_0 - x_0 < L(m).$$
Because $N$ is supported on $\Sigma$, the interval $[x_0, y_0]$ can be replaced with $[x_0, y_0] \cap \Sigma$. Every point of $\Sigma$ is contained in a generating band (Lemma \ref{gen-spect}), so it is not a loss to assume $[x_0, y_0]$ is contained in a generating band. Then, by the definition of $m$, there is a band $[x,y] \in \mathcal{G}_m$ containing it (notice that $y-x$ is comparable in size to $y_0 - x_0$, because of how $m$ is defined). We now have
$$N(y_0) - N(x_0) \le \lim_{k\to\infty} \frac{\#\sigma_{k+1,0} \cap [x,y]}{q_k}.$$

This ratio scales like $1/q_m$. Indeed, it is equal to $1/q_m$ at $k=m$. As $k$ increases to infinity, Lemma \ref{band-comb} shows that every band of $\sigma_{m+1,0}$ produces a roughly constant proportion of the bands that comprise $\sigma_{k+1,0}$. This means the share of bands in $\sigma_{k+1,0}$ produced by $[x,y]$ remains practically constant. So, pick $C$ so that $N(y_0) - N(x_0) \le C/q_m$.

Define $\gamma_k$ by
\begin{equation}\label{gamma-def}\gamma_k = \frac{\log C + \log q_k}{\log L(k+1)},\end{equation}
so that $C/q_m = L(m+1)^{\gamma_m}$, which is by the definiton of $m$ less than $(y-x)^{\gamma_m}$.

Take any
$$0 < \gamma < \liminf_{k\to\infty} \gamma_k = \frac{\log \beta}{L(b,\lambda)},$$
and choose $k_0$ so that $\gamma_k > \gamma$ for $k > k_0$. Put $\delta = L(k_0)$. Then if $a<b$ satisfy $b-a < \delta$, $N(b) - N(a) \le (b-a)^{\gamma_m}$, where $m$ is the integer corresponding to $[a,b]$. This is less than $(b-a)^\gamma$.

\end{proof}

\begin{proof} [Proof of Theorem \ref{opt-bb}]
    
Recall $U(k)$ from Lemma \ref{blub-bb}; $\log U$ goes to $-\infty$ roughly linearly. 

Let a supposed H\"older exponent $\tilde\gamma$ be given, satisfying the hypotheses of the theorem. Given $\delta$, pick $k_0$ so that $U(k_0) < \delta$ and
$$\gamma_m = \frac{-m\log \beta}{\log U(m)}$$
is less than $\tilde\gamma$ for all $m \ge k_0$. This is possible for all $\gamma > \log\beta/U(b,\lambda)$. Now, choose $[x,y] \in \sigma_{k_0}$ so that $y-x \ge U(k_0)$.

The exponent $\gamma_m$ is constructed so that $U(m)^{\gamma_m} = \beta^m$. We have already seen in the proof of \ref{hcont-bb} that
$$N(y) - N(x) = \lim_{k\to\infty} \frac{\#\sigma_k \cap [x,y]}{q_k} \simeq 1/q_{k_0}.$$
By construction, this is greater than $(y-x)^{\gamma_{k_0}} > (y-x)^\gamma$.
\end{proof}

\begin{proof} [Proof of Theorem \ref{hcont-d}]

We first prove that $N$ is H\"older continuous when $\overline{d}$ is finite. The proof of Theorem \ref{hcont-bb} may be followed until \ref{gamma-def}. By Lemma \ref{ub-arb}, $\log L(k)$ lies between two arithmetic sequences. And by Lemma \ref{qgrow}, the same goes for $\log q_k$. This implies $\liminf _{k\to\infty} \gamma_k > 0$, which is the necessary input to obtain H\"older continuity.

Now assume $\underline{d}=\infty$. We will follow the framework of the the proof of Theorem \ref{opt-bb} and see that the optimal H\"older exponent is zero. Let $U(k)$ stand for the sequence of upper bounds on band length obtained in Lemma \ref{nope-pd}; $\log U$ goes to $-\infty$ faster than any linear function. Also, recall that, by Lemma \ref{qgrow}, there is an $R$ so that $R^m$ grows faster than $q_m$.

Let a supposed H\"older exponent $\gamma >0 $ be given. Given $\delta$, pick $k_0$ so that $U(k_0) < \delta$ and
$$\gamma_m = \frac{-m\log R}{\log U(m)}$$
is less than $\gamma$ for all $m \ge k_0$. This is possible for all $\gamma > 0$ since $\lim_{m\to\infty} \gamma_m = 0$. Now, choose $[x,y] \in \sigma_{k_0}$ so that $y-x \ge U(k_0)$.

The exponent $\gamma_m$ is constructed so that $U(m)^{\gamma_m} = R^{-m}$. Again,
$$N(y) - N(x) = \lim_{k\to\infty} \frac{\#\sigma_k \cap [x,y]}{q_k} \simeq 1/q_{k_0}.$$
Because of the hypothesis on $\beta$, $1/q_{k_0} \ge R^{k_0} = U(k_0)^{\gamma_{k_0}}$. By construction, this is greater than $(y-x)^{\gamma_{k_0}} > (y-x)^\gamma$, proving that $N$ is not H\"older continuous.

\end{proof}

\section*{Acknowledgement}
The author is grateful to his advisor, David Damanik, for many helpful discussions and comments.

\end{document}

%% file: comb_bands.pdf_t
\begin{picture}(0,0)%
\includegraphics{comb_bands.pdf}%
\end{picture}%
\setlength{\unitlength}{4144sp}%
\begingroup\makeatletter\ifx\SetFigFont\undefined%
\gdef\SetFigFont#1#2#3#4#5{%
  \reset@font\fontsize{#1}{#2pt}%
  \fontfamily{#3}\fontseries{#4}\fontshape{#5}%
  \selectfont}%
\fi\endgroup%
\begin{picture}(6344,6794)(1329,-6383)
\put(3151,164){\makebox(0,0)[lb]{\smash{{\SetFigFont{12}{14.4}{\rmdefault}{\mddefault}{\updefault}{\color[rgb]{0,0,0}$(k,I)$}%
}}}}
\put(2971,-736){\makebox(0,0)[lb]{\smash{{\SetFigFont{12}{14.4}{\rmdefault}{\mddefault}{\updefault}{\color[rgb]{0,0,0}$(k+1,II)$}%
}}}}
\put(3466,-1636){\makebox(0,0)[lb]{\smash{{\SetFigFont{12}{14.4}{\rmdefault}{\mddefault}{\updefault}{\color[rgb]{0,0,0}$(k+2,I)$}%
}}}}
\put(6301,-2446){\makebox(0,0)[lb]{\smash{{\SetFigFont{12}{14.4}{\rmdefault}{\mddefault}{\updefault}{\color[rgb]{0,0,0}$(k+2,III)$}%
}}}}
\put(2746,-4336){\makebox(0,0)[lb]{\smash{{\SetFigFont{12}{14.4}{\rmdefault}{\mddefault}{\updefault}{\color[rgb]{0,0,0}$(k+2,III)$}%
}}}}
\put(6121,-5236){\makebox(0,0)[lb]{\smash{{\SetFigFont{12}{14.4}{\rmdefault}{\mddefault}{\updefault}{\color[rgb]{0,0,0}$(k+3,I)$}%
}}}}
\put(4996,-6181){\makebox(0,0)[lb]{\smash{{\SetFigFont{12}{14.4}{\rmdefault}{\mddefault}{\updefault}{\color[rgb]{0,0,0}$(k+3,III)$}%
}}}}
\end{picture}%